\documentclass[11pt,a4paper]{article}

\usepackage[utf8]{inputenc}
\usepackage[english]{babel}
\usepackage[T1]{fontenc}

\usepackage[affil-it]{authblk} 

\usepackage{amsmath}
\usepackage{amsfonts}
\usepackage{amssymb}
\usepackage{amsthm}

\usepackage[all]{xy} 

\usepackage{rotating}
\usepackage{multirow}
\usepackage{textcomp}

\usepackage{enumerate}
\usepackage[shortlabels]{enumitem}

\usepackage{algorithm} 
\usepackage[noend]{algorithmic}

\usepackage{graphicx}
\usepackage{makeidx,shortvrb,latexsym}
\usepackage{fancyhdr}
\usepackage{indentfirst}
\usepackage[usenames,dvipsnames]{pstricks}
\usepackage{epsfig}

\usepackage{hyperref} 
\hypersetup{
    colorlinks,
    citecolor=black,
    filecolor=black,
    linkcolor=black,
    urlcolor=black
}

\usepackage{verbatim}
\usepackage{listings} 
\usepackage{colortbl}
\usepackage{xcolor}
\usepackage{color}
\definecolor{mygreen}{rgb}{0,0.6,0}
\definecolor{mygray}{rgb}{0.5,0.5,0.5}
\definecolor{mymauve}{rgb}{0.58,0,0.82}
\theoremstyle{plain}
\newtheorem{theorem}{Theorem}

\theoremstyle{definition}
\newtheorem{definition}{Definition}

\theoremstyle{remark}
\newtheorem{remark}{Remark}
\newtheorem{note}{Note}

\title{An Extension of Cook's Elastic Cipher}
\author{Emanuele Bellini%
  \thanks{Electronic address: \texttt{eemanuele.bellini@gmail.com}; Corresponding author}}

\author{Guglielmo Morgari%
  \thanks{Electronic address: \texttt{guglielmo.morgari@telsy.it}}}
\author{Marco Coppola%
  \thanks{Electronic address: \texttt{marco.coppola@telsy.it}}}

\affil{Telsy Elettronica S.p.A.,\\ 
Corso Svizzera, 185, Torino (TO), Italy}

\date{Dated: \today}

\begin{document}
\maketitle
\begin{abstract}
Given a block cipher of length $L$ Cook's elastic cipher allows to encrypt messages of variable length from $L$ to $2L$.
Given some conditions on the key schedule, Cook's elastic cipher is secure against any key recovery attack if the underlying block cipher is, and it achieves complete diffusion in at most $q+1$ rounds if the underlying block cipher achieves it in $q$ rounds.
We extend Cook's construction inductively, obtaining the first elastic cipher for any message length greater than $L$ with the same properties of security as Cook's elastic cipher, and whose complexity of the encryption/decryption grows polynomially with the size of the message.
\end{abstract}
\section{Introduction}
To overcome the computational and space overheads arising from plaintext-padding techniques and some disadvantages of ciphertext stealing \cite{rogaway2012security} (such as altering how the mode is applied to the last two blocks and the need of decrypting the last block before the next-to-last block), variable input length block ciphers have been introduced in the past \cite{bellare1999construction,patel2005efficient}. 
Some work has been done on variable-length pseudorandom functions. 
The focus went from variable-length inputs with fixed length output as applicable to MACs and hash functions, \cite{an1999constructing,bellare1996pseudorandom,bernstein1999stretch,black2000cbc} 
to modes that work on multiples of the original block length \cite{halevi2003tweakable,halevi2004parallelizable}. 
There have also been ad-hoc attempts to design a variable-length block cipher from scratch \cite{reeds1992cryptosystem,Schroeppel1999hasty}. \\
All mentioned proposals for converting existing block ciphers into variable-length ones focused on treating a block cipher as a black box and combining it with other operations in what amounts to a mode of encryption. 
While such an approach allows the security of the variable-length block cipher to be defined in terms of the original block cipher, the resulting constructions require multiple applications of the original block cipher, making them computationally inefficient compared to padding. 
These methods may be valuable in providing modes of encryption that preserve the length of the data but they do not address how to design block ciphers to support variable-length blocks. 
Thus the concept of an \emph{elastic cipher} has been introduced by Cook et al. in \cite{cook2006elastic,cook2007elastic,cook2008methods,cook2009elastic}. 
Built on an existing block cipher, they provide a variable input length cipher whose length can be extended up to two times the original block cipher length.
To gain efficiency, the design of such a cipher exploits existing components of a cipher (e.g. the internal round) threating them as a black box. 
To prove security they show that the elastic version of a block cipher is secure against attacks that attempt to recover key bits if the original, fixed-length version of the cipher is secure against such attacks. 
Their method is unique in that they show how to convert such an attack on the elastic version directly into an attack on the original version with a polynomially related time complexity.\\
The elastic construction allows also for new modes of encryption \cite[\S 3]{cook2009modes}.\\
Here we provide an extension of Cook's elastic cipher, providing a still efficient construction with input of any length, offering the same security of the underlying block cipher. 
An advantage of such a solution is that when encrypting two identical and concatenated plaintexts of length twice the underlying block cipher size, the resulting ciphertext appears as a completely random string, instead of two identical concatenated random string. 
In fact, our construction can also be seen as an efficient mode of operation itself.\\

We present here our extension of Cook's elastic block cipher.
Her extension $E_1$ allows to construct an elastic cipher of length from $L$ to $2L$, where $L$ is the message length of the fixed input length block cipher $E_0$.
We start from $E_1$ and apply the elastic extension to $E_1$ itself, obtaining $E_2$ and we relate its security against key recovery attacks to that of $E_1$, and we repeat the process to build iteratively an elastic extension $E_n$ of $E_0$ which takes inputs of length $2^{n-1}L$ up to $2^nL$.\\ 
Section \ref{secCook} gives a brief description of Cook's idea recalling which are the key schedule requirements.
Section \ref{secCookExt} describes our extension, showing the algorithm to encrypt a plaintext of length $2^{n-1}L+y$, with $0\leq y\leq2^{n-1}L$.
Section \ref{secSecurity} is devoted to the security of $E_n$, first analyzing how the extension exploits the diffusion properties of the round function of $E_0$, and second giving a proof of how the security against key recovery attack of $E_n$ is reduced to that of $E_0$.
This is achieved following and extending Cook's proof for $E_1$.
\section{Cook's elastic cipher}\label{secCook}
\subsection{Preliminaries}
Let $n\in \mathbb{N}$, $n\geq1$.
Let $\{0,1\}^{\geq n}$ denote the set of all binary strings with length at least $n$.
The function $l(M)$ denotes the length of the string $M$.
A \emph{message/plaintext space} $\mathcal{M}$ is a nonempty subset of $\{0,1\}^{\geq n}$ for which $M\in\mathcal{M}$ implies that $M'\in\mathcal{M}$ for all $M'$ of the same length of $M$.
A \emph{ciphertext space} (or \emph{range}) $\mathcal{C}$ is a nonempty subset of $\{0,1\}^{\geq n}$.
From now on, $\mathcal{C=M}$.
A \emph{key space} $\mathcal{K}$ is a nonempty set together with a probability measure on that set.
A \emph{pseudorandom permutation} (PRP) with key space $\mathcal{K}$, message space $\mathcal{M}$ and range $\mathcal{C}$ is a set of permutations $F= \{F_K\mid K\in\mathcal{K}\}$ where each $F_K:\mathcal{M}\rightarrow\mathcal{C}$.
We usually write $F:\mathcal{K}\times\mathcal{M}\rightarrow\mathcal{C}$.
We assume that $l(F_K(M))$ depends only on $l(M)$.
A \emph{variable input length (VIL) cipher} is a PRP $F:\mathcal{K}\times \mathcal{M}\rightarrow \mathcal{M}$.
A \emph{block-cipher round} (or simply a \emph{round}) is a key-dependent permutation of the message/ciphertext space, precisely $R:\mathcal{K}\times\{0,1\}^n\rightarrow\{0,1\}^n$.
As usual we can consider $\{0,1\}^n$ to denote the vector space $(\mathbb{F}_2)^n$.
The number $n$ is called the \emph{block length}.
We consider rounds made of a round function (which is an invertible vectorial Boolean function indicated with the letter $C$, and that may be applied to just some of the round input bits, as in DES), followed by an affine map and a sum in $(\mathbb{F}_2)^n$ with the round key.
A \emph{block-cipher (BC)} is a cipher $F:\mathcal{K}\times\{0,1\}^n\rightarrow\{0,1\}^n$ obtained from the composition of $r$ rounds, whose keys are generated from the cipher key $k\in\mathcal{K}$ by a key schedule.
We write $F(k,x)=R_0(k_0(k),x_0)\circ...\circ R_{r}(k_{r}(k),x_{r})$ with $x\in\mathcal{M}$ and $k\in\mathcal{K}$.
\subsection{Description of the elastic cipher}
Call $E_0$ a block cipher with input length $L$, composed by $r_0$ rounds, where the $i$\-th round is indicated by $R_{0i}$, with $i=0,\ldots,r_0-1$.
\begin{definition}
 The \emph{cycle $C_0$ of the block cipher} $E_0$ is a vectorial Boolean function made of the least, over any key, number (called length) of consecutive rounds such that each bit of the cycle output is a function of at least two input bits.
\footnote{E.g., AES cycle coincides with its round; DES cycle is the composition of two consecutive round.}
\end{definition}
\begin{remark}
 Note that the cycle of a block cipher has been defined in \cite{cook2006elastic} as \emph{the sequence of steps in which all bits have been processed by the round function $R_0$}.
\end{remark}
We consider only ciphers that can be divided in cycles of the same length.
Suppose $E_0$ has $c_0$ cycles, and let $C_{0j}$ be $E_0$ $j$\-th cycle function ($j=1,\ldots,r_0/x$ where $x$ is the number of rounds per cycle).
Consider only block ciphers whose round function has a whitening as last operation (this is very common in real ciphers such as AES, Serpent DES, ...). \\
The whitenings after the cycles can be incorporated and extended to the full input length when applying the last whitening in the last round inside each cycle. \\
$E_0$ has $r_0=c_0\cdot x$ rounds.\\
The elastic version of $E_0$, which we call $E_1$, is widely described in \cite{cook2009elastic} with some different notations. 
We can define now $E_1$, the \emph{elastic extension} of $E_0$
\begin{definition}\label{defElastic}
\footnote{The construction in Def.\ref{defElastic} is given by Cook \cite{cook2006elastic}.}
An \emph{elastic cipher} $E_1$ of $E_0$ is a VIL cipher with message space $\{0,1\}^{L+y}$ with $0\leq y\leq L$, constructed from $E_0$ as follows:
\begin{itemize}
 \item the number of rounds of $E_1$ is $r_1=c_0+\lceil c_0\cdot y/L\rceil$
 \item $E_1$'s $i$\-th round $R_{1i}$ has one $E_0$'s cycle as round function, whose input are the leftmost $L$ bits and whose last sum with the key is expanded to the rightmost $y$ bit string $T$ that is not input into the $E_0$ cycle; the round function $C_{0j}$ is followed by a swap and a xor operation:
 \begin{enumerate}
  \item call $R$ the rightmost $y$ bit string which is the bit-wise sum of $T$ with $y$ key bits.
  \item choose a $y$ bit string $S$ from the output of the cycle and let its bits be the rightmost ones of $R_{1i}$ output;
  \item replace $S$ with $S\oplus T$
 \end{enumerate}
 \item add an initial and final whitening
 \item add an initial and final key-dependent permutation (as in \cite[page 24]{cook2007elastic}) between the initial and final whitening and the rounds.
\end{itemize}
\end{definition}
\subsection{The choice of the number of rounds in \texorpdfstring{$E_1$}{Lg}}
\label{secNumRounds}
The choice of $r_1$ helps us implementing the extension later.
In $E_0$ we have, by its definition, that to encrypt $L$ bits safely, they have to be \emph{influent} in $c_0$ cycles, 
where the expression \emph{influent in a cycle} means that a change in a bit of the input of $E_0$ may imply some changes in the output bits of that cycle, or said with other words, 
a input bit is influent in a cycle if it influences some bits of the output (a input bit influences an output bit if changing its value while keeping all other input bits fixed, 
implies a change in the output bit with probability greater than zero, taking the probability over all possible input values \cite[page 6]{cook2009elastic}).
When we add some extra bits to the input, say $y$, and we apply the function $C_0$ to the $L+y$ bit string, we have $y$ bits not influencing the cycle.
So, in the next cycle, we “mix” the $y$ bits left out with the input of the second cycle by xoring them, and letting out some other $y$ bits.
The problem now is that we will have some other $y$ bits remaining out of the cycle.
So, if we have $c_0$ cycles in $E_0$, and we only do the trick of xoring bits to new inputs, at the end of $c_0$ cycles we will have $c_0\cdot y$ bits which have not been influent in some of the cycles.
So we need to add $c_0\cdot y/L$ round to our construction.
\subsection{Key schedule requirements for the elastic cipher}
We give three requirements for the key schedule taken from \cite[page 153]{cook2006elastic} which can be satisfied if we use a pseudorandom generator (e.g.
RC4 as in \cite[page 24]{cook2009elastic}):
\emph{
\begin{enumerate}
  \item the key schedule should be a stand-alone algorithm that is usable to any BC;
  \item the expanded-key bits should be (or as close to) pseudorandom (as practical);
  \item the expanded-key rate for elastic block cipher should be a small multiple of the key expansion rate of a standard BC.
\end{enumerate}
}
\section{Extension of the elastic cipher}\label{secCookExt}
\subsection{Our proposal}
Our proposal is to consider an elastic cipher $E_1$ working with fixed length messages (as it was a block cipher), and apply iteratively the elastic extension to it.
We can define a cycle of an elastic cipher exactly as a cycle for a block cipher: 
\begin{definition}
 The \emph{cycle $C_{1i}$ of an elastic cipher $E_1$}, with $i=0,\ldots,r_1/2-1$, is a vectorial Boolean function made of the least, over any key, number of consecutive $E_1$'s rounds such that each bit of $C_{1i}$ output is a function at least two input bits.
\end{definition}
Now we are ready to construct $E_2$, $E_3$, ..., $E_n$, simply by iterating the previous process.
The first round function of $E_2$ looks like Fig.\ref{fig4}, whose notation is explained in the following lines:
\begin{itemize}
 \item the subscripts $A,B,Y$ indicate respectively the $L$ leftmost, the $L$ second leftmost and the $Y$ rightmost bits of the input $P$ or of the key $K$.
 \item $R_{ki}$ is the round function of $E_k$ used in round $i$ of $E_k$.
 \item $C_{ki}$ is the cycle function of $E_k$ used as round function in round $i$ of $E_{k+1}$.
 \item $K_X$ is a portion of bits of the key $K$, where $X=[i,\ldots,j]$ indicates a set of indexes.
In Fig.\ref{fig4} we have: $A=[0,\ldots,L-1]$, $B=[L,\ldots,2L-1]$, $Y=[2L,\ldots,2L+y-1]$.
 \item $P_A^i$ is the portion $A$ of the input of the $i+1$\-th round, or equivalently the output of the $i$\-th round.
Similarly for  $P_B^i$ and  $P_Y^i$.
\end{itemize}
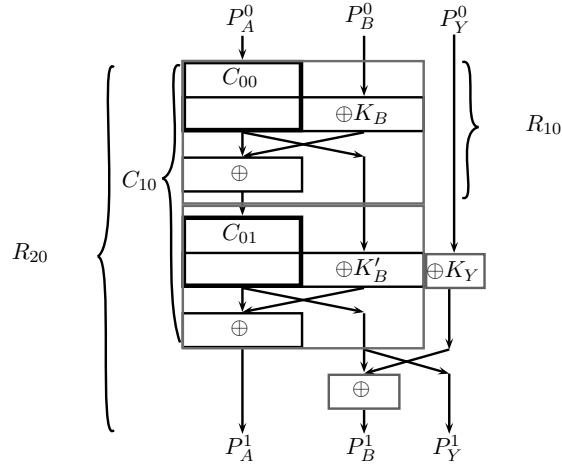
\begin{figure}
\centering
\scalebox{0.8} 
{
\begin{pspicture}(0,-3.7992187)(10.589063,3.7992187)
\definecolor{color513}{rgb}{0.4235294117647059,0.4235294117647059,0.4235294117647059}
\definecolor{color408}{rgb}{0.3764705882352941,0.37254901960784315,0.3843137254901961}
\psframe[linewidth=0.074,dimen=outer](5.51,2.8657813)(3.51,1.6657813)
\psframe[linewidth=0.04,dimen=outer](7.51,2.2657812)(3.51,1.6657813)
\psframe[linewidth=0.04,dimen=outer](5.51,1.2657813)(3.51,0.66578126)
\psline[linewidth=0.04cm,arrowsize=0.05291667cm 2.0,arrowlength=1.4,arrowinset=0.4]{<-}(4.51,2.8657813)(4.51,3.2657812)
\psline[linewidth=0.04cm,arrowsize=0.05291667cm 2.0,arrowlength=1.4,arrowinset=0.4]{<-}(6.51,2.2657812)(6.51,3.2657812)
\psline[linewidth=0.04cm,arrowsize=0.05291667cm 2.0,arrowlength=1.4,arrowinset=0.4]{<-}(4.51,1.2657813)(4.51,1.6657813)
\psline[linewidth=0.04cm,arrowsize=0.05291667cm 2.0,arrowlength=1.4,arrowinset=0.4]{<-}(4.51,1.2657813)(6.51,1.6657813)
\psline[linewidth=0.04cm,arrowsize=0.05291667cm 2.0,arrowlength=1.4,arrowinset=0.4]{<-}(6.51,1.2657813)(4.51,1.6657813)
\psbezier[linewidth=0.04](8.391898,1.7531185)(8.123966,1.7531185)(8.391898,2.9257812)(8.123966,2.8280594)
\psbezier[linewidth=0.04](8.391898,1.7531185)(7.99,1.6553966)(8.391898,0.4827338)(8.123966,0.58045566)
\usefont{T1}{ptm}{m}{n}
\rput(1.0245312,-0.32921875){$R_{20}$}
\usefont{T1}{ptm}{m}{n}
\rput(4.474531,2.5707812){$C_{00}$}
\usefont{T1}{ptm}{m}{n}
\rput(4.4645314,0.99078125){$\oplus$}
\usefont{T1}{ptm}{m}{n}
\rput(6.474531,1.9707812){$\oplus K_B$}
\usefont{T1}{ptm}{m}{n}
\rput(4.4845314,3.5507812){$P_A^0$}
\usefont{T1}{ptm}{m}{n}
\rput(6.4445314,3.5907812){$P_B^0$}
\usefont{T1}{ptm}{m}{n}
\rput(8.004531,3.5307813){$P_Y^0$}
\psframe[linewidth=0.074,dimen=outer](5.51,0.28578126)(3.51,-0.9142187)
\psframe[linewidth=0.04,dimen=outer](7.51,-0.31421876)(3.51,-0.9142187)
\psframe[linewidth=0.04,dimen=outer](5.51,-1.3142188)(3.51,-1.9142188)
\psline[linewidth=0.04cm,arrowsize=0.05291667cm 2.0,arrowlength=1.4,arrowinset=0.4]{<-}(4.51,0.28578126)(4.51,0.68578124)
\psline[linewidth=0.04cm,arrowsize=0.05291667cm 2.0,arrowlength=1.4,arrowinset=0.4]{<-}(6.51,-0.31421876)(6.51,1.2657813)
\psline[linewidth=0.04cm,arrowsize=0.05291667cm 2.0,arrowlength=1.4,arrowinset=0.4]{<-}(4.51,-1.3142188)(4.51,-0.9142187)
\psline[linewidth=0.04cm,arrowsize=0.05291667cm 2.0,arrowlength=1.4,arrowinset=0.4]{<-}(4.51,-1.3142188)(6.51,-0.9142187)
\psline[linewidth=0.04cm,arrowsize=0.05291667cm 2.0,arrowlength=1.4,arrowinset=0.4]{<-}(6.51,-1.3142188)(4.51,-0.9142187)
\psline[linewidth=0.04cm,arrowsize=0.05291667cm 2.0,arrowlength=1.4,arrowinset=0.4]{<-}(4.51,-3.3342187)(4.51,-1.9142188)
\psline[linewidth=0.04cm,arrowsize=0.05291667cm 2.0,arrowlength=1.4,arrowinset=0.4]{<-}(6.51,-2.3142188)(6.51,-1.3342187)
\usefont{T1}{ptm}{m}{n}
\rput(4.474531,-0.00921875){$C_{01}$}
\usefont{T1}{ptm}{m}{n}
\rput(4.4645314,-1.5892187){$\oplus$}
\usefont{T1}{ptm}{m}{n}
\rput(6.474531,-0.6092188){$\oplus K'_B$}
\usefont{T1}{ptm}{m}{n}
\rput(6.4645314,-2.5892189){$\oplus$}
\psframe[linewidth=0.04,linecolor=color408,dimen=outer](8.51,-0.33421874)(7.51,-0.93421876)
\usefont{T1}{ptm}{m}{n}
\rput(7.974531,-0.62921876){$\oplus K_Y$}
\psline[linewidth=0.04cm,arrowsize=0.05291667cm 2.0,arrowlength=1.4,arrowinset=0.4]{<-}(7.99,-0.33421874)(7.99,3.2857811)
\usefont{T1}{ptm}{m}{n}
\rput(4.4645314,-3.5692186){$P_A^1$}
\usefont{T1}{ptm}{m}{n}
\rput(6.4645314,-3.5692186){$P_B^1$}
\usefont{T1}{ptm}{m}{n}
\rput(7.884531,-3.5692186){$P_Y^1$}
\psline[linewidth=0.04cm,arrowsize=0.05291667cm 2.0,arrowlength=1.4,arrowinset=0.4]{<-}(6.51,-3.3342187)(6.51,-2.9342186)
\psline[linewidth=0.04cm,arrowsize=0.05291667cm 2.0,arrowlength=1.4,arrowinset=0.4]{<-}(7.91,-1.9342188)(7.91,-0.93421876)
\psline[linewidth=0.04cm,arrowsize=0.05291667cm 2.0,arrowlength=1.4,arrowinset=0.4]{<-}(6.51,-2.3342187)(7.91,-1.9342188)
\psline[linewidth=0.04cm,arrowsize=0.05291667cm 2.0,arrowlength=1.4,arrowinset=0.4]{<-}(7.91,-2.3342187)(6.51,-1.9342188)
\psline[linewidth=0.04cm,arrowsize=0.05291667cm 2.0,arrowlength=1.4,arrowinset=0.4]{<-}(7.91,-3.3342187)(7.91,-2.3342187)
\psframe[linewidth=0.04,linecolor=color513,dimen=outer](7.51,2.8657813)(3.51,0.46578124)
\psframe[linewidth=0.04,linecolor=color513,dimen=outer](7.51,0.46578124)(3.51,-1.9342188)
\psframe[linewidth=0.04,linecolor=color408,dimen=outer](7.11,-2.3342187)(5.91,-2.9342186)
\usefont{T1}{ptm}{m}{n}
\rput(9.464531,1.7907813){$R_{10}$}
\psbezier[linewidth=0.04](3.1917381,0.8977824)(3.4839127,0.8977824)(3.1917381,-0.9742165)(3.4839127,-1.8182166)
\psbezier[linewidth=0.04](3.1917381,0.8977824)(3.63,1.0537823)(3.1917381,2.9257812)(3.4839127,2.7697814)
\psbezier[linewidth=0.04](2.13,-0.39495495)(2.3979259,-0.39495495)(2.13,-3.534096)(2.3979259,-3.272501)
\psbezier[linewidth=0.04](2.13,-0.39495495)(2.5318887,-0.13335985)(2.13,3.0057812)(2.3979259,2.7441862)
\usefont{T1}{ptm}{m}{n}
\rput(2.8245313,0.87078124){$C_{10}$}
\end{pspicture} 
}
\caption{Details of the first round of $E_2$.
\label{fig4}}
\end{figure}
\subsubsection{Some terminology}
Call $E_n$ an elastic cipher of $E_0$ of \emph{level} $n$.
The integer $n$ is also called the level of the extension $E_n$.
$E_n$ is able to encrypt blocks from length $2^{n-1}L$ to $2^nL$.
We call $R_n$ and $C_n$ respectively the round and the cycle functions of $E_n$.
The integers $r_n$ and $c_n$ indicates respectively the number of round and cycles in $E_n$, while
$r_m/E_n$ with $n>m$, indicates the number of rounds of $E_m$ used in $E_n$.
\begin{remark}
 The level increases by 1 when the message length doubles, and in general the level increases by $k$ when the length passes from $2^nL$ to $2^{n+k}L$, which means the level increases logarithmically in respect to the message length.
\end{remark}
\begin{remark}
 Once $n$ is fixed, only $E_n$ takes variable inputs and has initial and final key-dependent permutations, all other $E_i$, with $i<n$, take as input messages of fixed length $2^iL$ and do not have initial and final key-dependent permutation.
\end{remark}
\subsection{Elastic cipher extension of level \texorpdfstring{$n$}{Lg}}
Here we give explicitly the steps to encrypt a message of any length greater than $L$ using an extended elastic cipher of a certain level, obtained from a FIL block cipher $E_0$ that works on messages of length $L$.
This algorithms are the definition of $E_n$.
Furthermore we give some theorems 
to justify the formulas used.
We divide the encrypting algorithm in two parts: initialization of $E_n$ and encryption step.
For brevity, in this article, we will omit the algorithm for decryption and $CycleFunction^{-1}$.
\subsubsection{Initialization of \texorpdfstring{$E_n$}{Lg}}
The initialization of $E_n$ is described in Algorithm \ref{alg1}.
\begin{algorithm}[h]
\caption{Initialization of $E_n$.
We get a message $P$ as input and return the parameters needed to encrypt $P$.}
\label{alg1}
The following algorithm initializes the elastic cipher $E_n$, given a message $P$ of given length.\\
\small
INPUT: $P$.\\
OUTPUT: the level $n$, the exponent $y$, the number of rounds $r_n,l(K)$, the key length, $K$ a random vector of length $l(K)$.
\begin{enumerate}
 \item Determine the level $n$ of the elastic extension as the least power of 2 such that $2^nL\geq l(P)$
 \item Determine $y=l(P)-2^{n-1}L$ // $l$ is the length function.
 \item Determine $r_n=c_0+\lceil c_0\frac{y}{2^{n-1}L}\rceil$ \footnote{The integer $r_n$ has the property that it is the least number of $E_n$ rounds such that, in $E_n$, each input bit influences $E_0$ round function the same number of times it influences the same round function in $E_0$.}
 \item Determine the key $K$:
    \subitem Determine the key length: // (see Th.\ref{thKeyLength})\\
    $l(K)=\{y+2^{n-1}[L\cdot(n-1)+l_{K_R}(0)\cdot x]\}\cdot r_n+2^nL+2y+K_{perm}$
    \subitem Fill a $l(K)$ string with random bits // e.g.
using a pseudorandom generator
\end{enumerate}
\normalsize
\end{algorithm}

\subsubsection{Encryption function \texorpdfstring{$E_n$}{Lg}}
The steps to encrypt a message $P$ with $E_n$ initialized are described in Algorithm \ref{alg2}.\\
\begin{algorithm}[h]
\caption{Encryption function}
\label{alg2}
\small
INPUT: $P,K,y,n,r_n$.\\
OUTPUT: $P$ encrypted.
\begin{enumerate}
 \item $k=0$ // $k$ is a pointer to the first bit not used of the key.
 \item $P=P\oplus K_{[k,...,k+2^{n-1}L+y-1]}$ //Initial whitening
 \item $k=k+2^{n-1}L+y$
 \item $P=perm(P,K_{[k,...,k+l(K_{perm})-1]})$ //Key dependent permutation, \cite[p.
24]{cook2007elastic}.
 \item $k=k+l(K_{perm})$ // move the pointer $k$.
 \item $j=0$
 \item FOR $i=0,\ldots,r_n-1$
  \begin{enumerate}
   \item $P_{[0,...,2^{n-1}L-1]}=CycleFunction(P_{[0,...,2^{n-1}L-1]},n-1)$ // cycle of $E_{n-1}$
   \item $P_{[2^{n-1}L,...,2^{n-1}L+y-1]}=P_{[2^{n-1}L,...,2^{n-1}L+y-1]}\oplus K_{[k,...,k+y-1]}$
   //Add key
   \item $k=k+y$
   \item $T_{[0,...,y-1]}=P_{[j,...,j+y-1]}$ //$T$ is a temporary value for the swap steps
   \item $P_{[j,...,j+y-1]}=P_{[j,...,j+y-1]}\oplus P_{[2^{n-1}L,...,2^{n-1}L+y-1]}$\footnote{Exoring with the bits left out from the $CycleFunction$.
$j$ fixed would not influence the proof of security.
}
   \item $P_{[2^{n-1}L,...,2^{n-1}L+y-1]}=T_{[0,...,y-1]}$ // (d)-(e)-(f) are the swap/xor steps.
   \item $j=j+1 \mod 2^{n-1}L$
  \end{enumerate}
 \item $P=perm(P,K_{[k,...,k+l(K_{perm})-1]})$
 \item $k=k+l(K_{perm})$
 \item $P=P\oplus K_{[k,...,k+2^{n-1}L+y-1]}$
 \item RETURN $P$
\end{enumerate}
\normalsize
\end{algorithm}
\subsubsection{\texorpdfstring{$CycleFuntion(M,m)$}{Lg}}
The cycle of $E_m$, used in the encryption function, is described in Algorithm \ref{alg4}.\\
\begin{algorithm}[h]
\caption{$CycleFuntion(M,m)$.}
\label{alg4}
\small
INPUT: the bit string $M$, the level $m$ of the cycle function ($M$ has length $2^mL$).\\
OUTPUT: the $2^mL$ bit string $M$ processed by $CycleFunction$.
\begin{enumerate}
 \item IF $m=0$ THEN
  \begin{enumerate}
    \item FOR $i=1,\ldots,x$ DO
	\begin{enumerate}
	\item $M=R_0(M,K_{k,...,k+l_{K_R}(0)-1})$
	\item $k=k+l_{K_R}(0)$
	\end{enumerate}
    \item RETURN M
  \end{enumerate}
 \item IF $m>0$ THEN
  \begin{enumerate}
   \item $A=M_{[0,...,2^{m-1}L-1]}$
   \item $B=M_{[2^{m-1}L,...,2^mL-1]}$
   \item FOR $i=1,2$ DO
   \begin{enumerate}
    \item $A=CycleFunction(A,m-1)$
    \item $B=B\oplus K_{[k,...,k+2^mL-1]}$
    \item $k=k+2^mL$
    \item $T=A$
    \item $A=A\oplus B$
    \item $B=T$
   \end{enumerate}
   \item RETURN $A \arrowvert\arrowvert B$
  \end{enumerate}
\end{enumerate}
\normalsize
\end{algorithm}
\subsubsection{On the number of rounds and key length}
Notice that $r_n$ does not represent the number of rounds of level $m$ in $E_n$, but the number of rounds of level $n$ in $E_n$.
To know how many times the cycle function (or equivalently the round function) of $E_m$ is used in $E_n$, then we need the following:
\begin{theorem} [Number of Rounds $R_m$ in $E_n$]\label{thNumRoundsR0}
The number $r_{m,n}$ of rounds of level $m$ in $E_n$ is:
\begin{equation}
 r_m/E_n = r_n\cdot2^{n-m} \,.
\end{equation}
In particular\footnote{The number of rounds of level 0 used in the elastic increases exponentially with $n$, which means polynomially (in our case linearly) with the input length $2^nL$.} if $m=0$:
\begin{equation}
 r_0/E_n = r_n\cdot2^{n-1}\cdot x\leq r_0\cdot 2^{n-1}\cdot x \,.
\end{equation}
\end{theorem}
\begin{proof}
This is because one round of $E_n$ contains one cycle of $E_{n-1}$, which is made of two rounds of $E_{n-1}$(plus a whitening and a swap), and each of them contains 2 rounds of $E_{n-2}$(plus a whitening and a swap), and so on; until one round of $E_1$ is made of $x$ rounds of $E_0$(plus a whitening and a swap).
In formulas:
\begin{align*}
r_0/E_n &= r_n\cdot2^{n-1}\cdot x = \\
& = c_{n-1}+\lceil c_{n-1}\frac{y}{2^{n-1}L}\rceil 2^{n-1}\cdot x \le \\
& \leq 2 \cdot c_0 \cdot2^{n-1} \cdot x =\\
& = c_0 \cdot 2^n \cdot x = \\ 
& = r_0 \cdot 2^{n-1} \cdot x \,.\\
\end{align*}
\end{proof}
The following theorem, beside giving an explicit formula for the key bits needed for encryption shows that Cook’s third requirement for the key schedule is satisfied because the key expansion rate grows polynomially with the input length.
Let’s call $l_{K_C}(n)$ the function that gives the number of bits needed for the key of the cycle of level $n$, or simply the key length of cycle $C_n$.
\begin{theorem}[Key Length]\label{thKeyLength}
 If we consider $E_n$ with initial and final whitening and key-dependent permutation, the number of bits needed as key bits is:
\begin{equation}
 l(K)=\{y+2^{n-1}[L\cdot(n-1)+l_{K_R}(0)\cdot x]\}\cdot r_n+2^nL+2y+K_{perm}\,.
\end{equation}
\end{theorem}
\begin{proof}
Notice that if our extended cipher is of level $n$, then our function is defined only for level $<n$, because those are the ones with fixed input length.
The function may be defined recursively in the following way: 
\begin{equation}
l_{K_C}(n)=\left\{
	    \begin{array}{rl}
		2^nL+2l_{K_C}(n-1) & \text{if } n>0\\
		l_{K_C}(0)=l_{K_R}(0)\cdot x & \text{if } n=0
	    \end{array} \right.
\end{equation}
We are now ready to prove the formula.
We see that for each cycle key we need $y$ bits for “extra” whitening and $l_{K_C}(n-1)$ bits.
Considering also initial and final whitenings and key-dependent permutations we have:
\begin{equation} 
l(K)=[y +l_{K_C}(n-1) +]\cdot r_n +2(2^{n-1}L+y)+K_{perm}\,.
\end{equation}
If we write $l_{K_C}(n-1)$ in a compact form we obtain:
\begin{align*}
l_{K_C}(n-1) 
& = 2^nL+2l_{K_C}(n-2) = \\
& = 2^nL+2(2^{n-1}L+2l_{K_C}(n-3)) = \\
& = 2^nL+2^nL+2^2l_{K_C}(n-3) = \\
& = \ldots = \\
& = \underbrace{2^nL+...+2^nL}_{n-1}\text{times}+2^{n-1\cdot l_{K_C}(0)} = \\
& = 2^{n-1}L(n-1)+2^{n-1}\cdot l_{K_R}(0)\cdot x = \\
& = 2^{n-1}[L\cdot(n-1)+l_{K_R}(0)\cdot x]\,.
\end{align*}
\end{proof}

\section{Security of the elastic cipher extension}\label{secSecurity}
\subsection{Diffusion}
We recall that if every input bit to a $b$-bit block cipher influences the value in all $b$ bits after $q$ rounds, then the block cipher is said to have \emph{complete diffusion} in $q$ rounds.
Complete diffusion does not imply security and is not the same as diffusion in an ideal
block cipher where changing a single bit of input will cause each individual bit of output to
change with $1/2$ probability. We refer to this case as \emph{ideal diffusion}. 
In complete diffusion, the probability each individual bit of output changes must only be $> 0\%$ and may be $100\%$.\\
In her paper, \cite[page 8]{cook2007elastic}, Cook shows that if complete diffusion occurs after $q$ cycles in $E_0$ (the fixed length block cipher), then it occurs after at most $q+1$ rounds in $E_1$ (the elastic version of $E_0$).
It is possible to extend this proof for $E_n$.
\begin{theorem}[Complete Diffusion]\label{thComplDiff}
If complete diffusion occurs after $q$ cycles in $E_{n-1}$ (an elastic cipher working with length message $2^{n-1}L$), then it occurs after at most $q+1$ rounds in $E_n$(the elastic version of $E_{n-1}$).
\end{theorem}
\begin{proof}
 According to \cite[page 6]{cook2007elastic}, complete diffusion is achieved in $q$ rounds \emph{if every single bit input to a $l$-bit block cipher influences the value in all $l$ bits after $q$ rounds}, 
where we say that a $q$ rounds input bit $b_i$ influences a $q$ rounds output bit $b_j$ if \emph{changing $b_i$, 
while holding all other $l-1$ input bits constant, causes $b_j$ to change with probability $>0$, when the probability is taken over all possible values of other input bits and the key bits are held constant}.\\
First notice that whitening does not impact diffusion, so, for the purposes of our proof, we can view the elastic expansion without them.\\
We can see that each of the first $2^{n-1}$ input bits $b_j$ that influences a bit in position $j$ of the output of $C_{n-1}$ will influence one (or more) bit in position $j$ of the output of $R_n$, and if $b_i$ influences a bit xored after the cycle, then $b_i$ will also influence a output bit of $R_n$ in position between $2^{n-1}L$ and $2^{n-1}L+y-1$.\\
So, if we have complete diffusion in $E_{n-1}$ after $q$ rounds, then, in $E_n$, the first $2^{n-1}L$ input bits will keep influencing all $2^nL+y$ output bits after $q$ rounds.\\
In $E_n$, the rightmost $y$ bits of the input do not influence the first cycle, 
but because of the swap step they influence the second cycle and since then they will influence all the cycles and will reach complete diffusion after $q+1$ cycles.
Figure \ref{fig9} shows the situation.
\begin{figure}
\centering
\scalebox{0.8} 
{
\begin{pspicture}(0,-5.163594)(9.385625,5.163594)
\definecolor{color274b}{rgb}{0.42745098039215684,0.42745098039215684,0.42745098039215684}
\definecolor{color274}{rgb}{0.3803921568627451,0.3803921568627451,0.3803921568627451}
\definecolor{color590}{rgb}{0.00392156862745098,0.00392156862745098,0.00392156862745098}
\psframe[linewidth=0.04,dimen=outer](6.2,3.4257812)(0.0,2.6257813)
\psframe[linewidth=0.04,dimen=outer](6.2,2.2257812)(0.0,1.4257812)
\psframe[linewidth=0.04,dimen=outer](6.2,-1.3742187)(0.0,-2.1742187)
\psframe[linewidth=0.04,dimen=outer](6.2,-2.5742188)(0.0,-3.3742187)
\psframe[linewidth=0.04,dimen=outer](6.2,4.2257814)(0.0,3.8257813)
\psframe[linewidth=0.04,dimen=outer](8.0,4.2257814)(6.2,3.8257813)
\psframe[linewidth=0.04,dimen=outer](6.2,-4.5742188)(0.0,-4.974219)
\psframe[linewidth=0.04,dimen=outer](8.0,-4.5742188)(6.2,-4.974219)
\psframe[linewidth=0.04,dimen=outer](4.8,-2.5742188)(3.0,-3.3742187)
\psline[linewidth=0.04cm](6.8,4.2257814)(6.8,2.6257813)
\psline[linewidth=0.04cm](6.8,2.6257813)(5.2,2.2257812)
\psline[linewidth=0.04cm](5.2,2.2257812)(6.0,-3.3742187)
\psline[linewidth=0.04cm](5.2,2.2257812)(4.6,-3.3742187)
\psline[linewidth=0.04cm](5.2,2.2257812)(3.2,-3.3742187)
\psline[linewidth=0.04cm](5.2,2.2257812)(0.6,-3.3742187)
\psline[linewidth=0.04cm](5.2,2.2257812)(0.2,-3.3742187)
\psline[linewidth=0.04cm](0.2,-3.3742187)(0.2,-4.5742188)
\psline[linewidth=0.04cm](0.6,-3.3742187)(0.6,-4.5742188)
\psline[linewidth=0.04cm](3.2,-3.3742187)(3.2,-4.5742188)
\psline[linewidth=0.04cm](4.6,-3.3742187)(4.6,-4.5742188)
\psline[linewidth=0.04cm](6.0,-3.3742187)(6.0,-4.5742188)
\psline[linewidth=0.04cm](3.2,-3.3742187)(6.4,-4.5742188)
\psline[linewidth=0.04cm](4.6,-3.3742187)(7.8,-4.5742188)
\psline[linewidth=0.04cm,linestyle=dashed,dash=0.16cm 0.16cm](5.2,2.2257812)(3.8,-3.3742187)
\psline[linewidth=0.04cm,linestyle=dashed,dash=0.16cm 0.16cm](3.8,-3.3742187)(3.8,-4.5742188)
\psline[linewidth=0.04cm,linestyle=dashed,dash=0.16cm 0.16cm](3.8,-3.3742187)(6.8,-4.5742188)
\psline[linewidth=0.04cm,linecolor=red,fillcolor=color274b](1.0,4.2257814)(1.0,3.4257812)
\psline[linewidth=0.04cm,linecolor=red,fillcolor=color274b](1.0,3.4257812)(0.2,-2.5742188)
\psline[linewidth=0.04cm,linecolor=red,fillcolor=color274b](1.0,3.4257812)(0.6,-2.5742188)
\psline[linewidth=0.04cm,linecolor=red,fillcolor=color274b](1.0,3.4257812)(6.0,-2.5742188)
\psline[linewidth=0.04cm,linecolor=red,fillcolor=color274b](6.0,-2.5742188)(6.0,-3.3742187)
\psline[linewidth=0.04cm,linecolor=red,fillcolor=color274b](0.6,-2.5742188)(0.6,-3.3742187)
\psline[linewidth=0.04cm,linecolor=red,fillcolor=color274b](0.2,-2.5742188)(0.2,-3.3742187)
\psline[linewidth=0.04cm,linecolor=red,fillcolor=color274b,linestyle=dashed,dash=0.16cm 0.16cm](1.0,3.4257812)(3.8,-2.5742188)
\psline[linewidth=0.04cm,linecolor=red,fillcolor=color274b,linestyle=dashed,dash=0.16cm 0.16cm](3.8,-2.5742188)(3.8,-3.3742187)
\psarc[linewidth=0.04,linecolor=color590,fillcolor=color274b](8.0,3.0257812){0.4}{0.0}{90.0}
\rput{-90.0}(10.974218,5.025781){\psarc[linewidth=0.04,linecolor=color590,fillcolor=color274b](8.0,-2.9742188){0.4}{0.0}{90.0}}
\rput{-180.0}(17.6,0.8515625){\psarc[linewidth=0.04,linecolor=color590,fillcolor=color274b](8.8,0.42578125){0.4}{0.0}{90.0}}
\rput{-270.0}(8.425781,-9.174219){\psarc[linewidth=0.04,linecolor=color590,fillcolor=color274b](8.8,-0.37421876){0.4}{0.0}{90.0}}
\psline[linewidth=0.04cm,linecolor=color590,fillcolor=color274b](8.4,3.0257812)(8.4,0.42578125)
\psline[linewidth=0.04cm,linecolor=color590,fillcolor=color274b](8.4,-0.37421876)(8.4,-2.9742188)
\usefont{T1}{ptm}{m}{n}
\rput{-90.0}(9.166407,9.166407){\rput(9.163593,0.01578125){Application of $R_n$ (and hence $C_{n-1}$) $q+1$ times}}
\end{pspicture} 
}
\caption{One-Bit Diffusion in the Elastic Network.
\label{fig9}}
\end{figure}
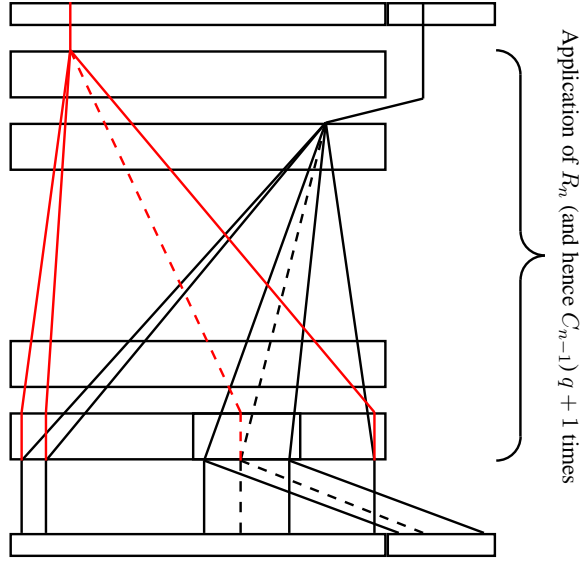
\end{proof}

\subsubsection{On diffusion probabilities and ideal diffusion}
We now illustrate how the elastic expansion impacts on the probabilities that a bit of the output of a round function changes, given that one input bit changed while keeping all other input bits fixed, and with constant key bits.\\
It is possible to generalize the previous theorem also in the case of ideal diffusion.
\begin{theorem}[Ideal Diffusion]\label{thIdealDiff}
If ideal diffusion occurs after $q$ cycles in $E_{n-1}$ (an elastic cipher working with length message $2^{n-1}L$), then it occurs after at most $q+1$ rounds in $E_n$(the elastic version of $E_{n-1}$).
\end{theorem}
\begin{proof}
Consider a black box $f$ that works with a $L$ bit string (which represents our cycle function).
We call the input bit string $P^{in}$ and the output bit string $P^{out}$, and we denote a single input bit (in position $i$) $P_i^{in}$ and a single output bit (in position $j$) $P_i^{out}$, where $0\leq i,j\leq L-1$.
Consider the following experiment $A$:
\small
\begin{enumerate}
 \item Initialize four counters $n_{00}=0,n_{01}=0,n_{10}=0,n_{11}=0$.
 \item Choose position $i,j$ for an input bit $P_i^{in}$ and an output bit $P_j^{out}$.
 \item Fix all other input bit to some value.
 \item Fix $P_i^{in}=0$, calculate $f(P^{in})$ and take note of the value assumed by $P_j^{out}$.
 \item Fix $P_i^{in}=1$, calculate $f(P^{in})$ and take note of the value assumed by $P_j^{out}$.
 \item If $P_j^{out}$ was equal to 0 both in point 4 and 5 then $n_{00}=n_{00}+1$.\\
       If $P_j^{out}$ was equal to 0 both in point 4 and 5 then $n_{11}=n_{11}+1$.\\
       If $P_j^{out}$ was equal to 0 in point 4 and to 1 in point 5 then $n_{01}=n_{01}+1$.\\
       If $P_j^{out}$ was equal to 1 in point 4 and to 0 in point 5 then $n_{10}=n_{10}+1$.
 \item Repeat the experiment for all other $2^{L-1}$ possible values of the input.
\end{enumerate}
\normalsize
Denote the number of times we had a change in the bit $P_j^{out}$ with $n_c=n_{01}+n_{10}$, and denote the number of times the value in $j$\-th position remained equal with $n_e=2^{L-1}-n_c=n_{00}+n_{11}$.\\
At the end of the experiment, changing the bit in position of the input produces a change in the bit in position $j$ of the output with probability $p_c=\frac{n_c}{2^{L-1}}$ and it does not produce a change with probability $p_e=\frac{n_e}{2^{L-1}}$.\\
The ideal situation for a cycle function is when $n_{00}=n_{11}=n_{01}=n_{10}=\frac{2^{L-1}}{4}$.\\
Now consider a second experiment $B$, where we have the same black box $f$ as in experiment $A$, and where we add some new steps, which are the elastic steps that influence diffusion, namely the xor and swap steps,
see in Fig. \ref{fig10} $(0\leq k<y)$.
\begin{figure}
\centering
\scalebox{1} 
{
\begin{pspicture}(0,-1.6792188)(9.449062,1.6792188)
\psframe[linewidth=0.04,dimen=outer](6.57,0.8257812)(3.37,0.22578125)
\psframe[linewidth=0.04,dimen=outer](4.97,-0.17421874)(4.17,-0.77421874)
\psline[linewidth=0.04cm,arrowsize=0.05291667cm 2.0,arrowlength=1.4,arrowinset=0.4]{<-}(4.97,0.8257812)(4.97,1.2257812)
\psline[linewidth=0.04cm,arrowsize=0.05291667cm 2.0,arrowlength=1.4,arrowinset=0.4]{<-}(4.57,-0.17421874)(4.57,0.22578125)
\psline[linewidth=0.04cm,arrowsize=0.05291667cm 2.0,arrowlength=1.4,arrowinset=0.4]{<-}(4.57,-0.17421874)(6.97,0.22578125)
\psline[linewidth=0.04cm,arrowsize=0.05291667cm 2.0,arrowlength=1.4,arrowinset=0.4]{<-}(6.97,0.22578125)(6.97,1.2257812)
\psline[linewidth=0.04cm,arrowsize=0.05291667cm 2.0,arrowlength=1.4,arrowinset=0.4]{<-}(6.97,-0.17421874)(4.57,0.22578125)
\psline[linewidth=0.04cm,arrowsize=0.05291667cm 2.0,arrowlength=1.4,arrowinset=0.4]{<-}(6.97,-1.1742188)(6.97,-0.17421874)
\psline[linewidth=0.04cm,arrowsize=0.05291667cm 2.0,arrowlength=1.4,arrowinset=0.4]{<-}(4.57,-1.1742188)(4.57,-0.77421874)
\usefont{T1}{ptm}{m}{n}
\rput(4.9245315,1.4507812){$P_i^{in}$}
\usefont{T1}{ptm}{m}{n}
\rput(4.5445313,-0.46921876){$\oplus$}
\usefont{T1}{ptm}{m}{n}
\rput(5.0445314,0.51078125){$f$}
\usefont{T1}{ptm}{m}{n}
\rput(6.9445314,1.4707812){$P_{k+2^{L-1}}^{in}$}
\usefont{T1}{ptm}{m}{n}
\rput(7.094531,-1.3892188){$P_{k+2^{L-1}}^{out}$}
\usefont{T1}{ptm}{m}{n}
\rput(4.3145313,-1.4492188){$P_j^{out}=f(P_i^{in})+P_{k+2^{L-1}}^{in}$}
\end{pspicture} 
}
\caption{One-Bit Diffusion in One Round.
\label{fig10}}
\end{figure}
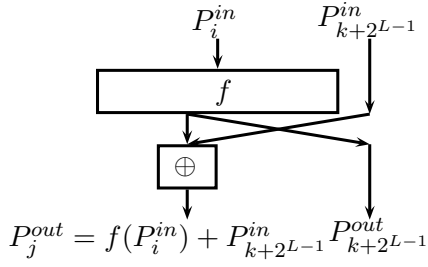
Now the input bit string has $L+y$ bits.\\
We want to repeat the experiment with this new scheme to see how the elastic expansion acts on the probability of a change in one single output bit given a variation in one single input bit.\\
First consider the output of $f$.
We observe that repeating the experiment $2^{L+y-1}$ times the output string right after the black box \footnote{we will refer to it as $f(P_i^{in})$, and $f(P_i^{in})_j$ will indicate its $j$\-th bit} will give $n_{rs}^B=n_{rs}^A\cdot 2^y$, where $n_{rs}^X$ denotes the number of times the $j$\-th output bit passed from $r$ to $s$ in the experiment $X$, with $r,s\in\{0,1\}$.
This is because we are actually repeating the experiment $A$ $2^y$ times, with the leftmost $y$ bits not influencing the function $f$.
This gives that $p_c^B=\frac{n_c^B}{2^{L+y-1}}=\frac{n_{01}^B+n_{10}^B}{2^{L+y-1}}=\frac{n_{01}^A2^y+n_{10}^A2^y}{2^{L+y-1}}=\frac{n_{01}^A+n_{10}^A}{2^{L-1}}=p_c^A$.\\
Now consider the output of the entire experiments.
Because last $y$ bits do not influence the black box function, they do not influence the part of the final output of the experiment $B$ that does not come from the output of the last xor.
So we are interested in analyzing the bits of the output in those position $j$ influenced by the xor step.\\
Through the entire experiment $B$, the bit $f(P_i^{in})$ will be xored half of the times with 0 and half of the times with 1 (which come from the rightmost $y$ bits).
This means that if we counted, for instance, $n_{00}^A\cdot 2^y$ times for the $j$\-th output bit of $f$ not changing from 0, after the xor we will count $\frac{n_{00}^A\cdot2^y}{2}$ times for the $j$\-th output bit of the xor not changing from 0.
With the same reasoning we obtain the following scheme:\\
\begin{align}
  \nonumber
  n_{00}^B=(\frac{n_{00}^A}{2}+\frac{n_{01}^A}{2})2^y,
  n_{11}^B=(\frac{n_{11}^A}{2}+\frac{n_{10}^A}{2})2^y,\\
  n_{01}^B=(\frac{n_{01}^A}{2}+\frac{n_{00}^A}{2})2^y,
  n_{10}^B=(\frac{n_{10}^A}{2}+\frac{n_{11}^A}{2})2^y
\end{align}
The probability a bit of the output (in position $j$) changed given a change of one bit in the input (in position $i$) is:
\begin{equation}
 p^B=\frac{n_{01}^B+n_{10}^B}{2^{L+y-1}}=
 \frac{(\frac{n_{01}^A}{2}+\frac{n_{00}^A}{2})2^y+(\frac{n_{10}^A}{2}+\frac{n_{11}^A}{2})2^y}{2^{L+y-1}}
  =\frac{1}{2}\cdot\frac{n_e^A+n_c^A}{2^{L-1}}
 =\frac{1}{2} 
\end{equation}
\noindent And similarly $q_c^B=1/2$.\\
This means that $P_j^{out}$ (when $j$ is a position included between the positions coming from the xor output) changes with probability $p^B=50\%$, given a variation of $P_i^{in}$.\\
Reasoning as in Theorem \ref{thComplDiff}, and with the last results, we proved the theorem.
\end{proof}
\begin{note}
 The fact that after the xor each bit changes with probability $1/2$ may be exploited as a weakness,
because an adversary could distinguish a block cipher from the elastic cipher if the elastic round function $R_0$ does not achieve ideal diffusion;
it would be sufficient to implement the experiment we described for a reasonable amount of input messages and see if there is an output bit sequence that change with probability $1/2$, while all other output bit change with probability significantly different from $1/2$.
\end{note}

\subsection{Security against key recovery}
In her paper, \cite[page 17]{cook2007elastic}, Cook proves the following theorem.
\begin{theorem}[Cook]\label{thCook}
Given a fixed-length block cipher, $E_0$, that works on $L$-bit blocks and its elastic version, 
$E_1$, that works on $(L+y)$-bit blocks, where $0\leq y\leq L$, if there exists an attack, $A_1/E_1$, on $E_1$ that allows the round keys to be determined for $r$ consecutive rounds of $E_1$ using polynomial (in $L$ and $r$) time and memory, 
then there exists an attack on $E_0$ with $r$ cycles that finds the expanded key for $E_0$ and that uses polynomial (in $L$ and $r$) many resources as $A_1/E_1$, assuming there are no message-dependent expanded key, 
meaning any expanded-key bits utilized in $E_0$ depend only on the key and do not vary across plaintext or ciphertext inputs.
\end{theorem}
We see Theorem\ref{thCook} as the base step of a proof by induction, and so we prove the following.
We now establish that $E_n$ is secure against any key recovery attack\footnote{In her paper, \cite[page 17]{cook2007elastic}, Cook proves this theorem for an elastic cipher $E_1$ applied to $E_0$ providing us the base step for our proof by induction.} if so is $E_0$.
\begin{theorem}[Security Against Key Recovery]\label{thSecKeyRec}
Given an elastic cipher, $E_{n-1}$ of level $n-1$ (without initial and final whitening and key-dependent permutation), 
working on $2^{n-1}L$-bit blocks and its elastic version, $E_n$, that works on $(2^{n-1}L+y)$-bit blocks, where $0\leq y\leq 2^{n-1}L$, if there exists an attack, $\mathcal{A}_n$, on $E_n$ that allows the round keys to be determined for $r$ consecutive rounds of $E_n$ using $t_{\mathcal{A}_n}$ operation, 
then there exists an attack $\mathcal{A}_{n-1}$ on $E_{n-1}$ with $r$ cycles that finds the expanded key for $E_{n-1}$ and that uses $t_{\mathcal{A}_{n-1}}<O(sr^2+rt_{\mathcal{A}_n})$, assuming there are no message-dependent expanded key, meaning any expanded-key bits utilized in $E_{n-1}$ depend only on the key and do not vary across plaintext or ciphertext inputs.
In particular, if $\mathcal{A}_n$ is polynomial then $\mathcal{A}_{n-1}$ is polynomial.
\end{theorem}

\begin{proof}
We show how the attack $\mathcal{A}_n$ can be converted to an attack to $E_{n-1}$ in polynomial time.\\
Suppose we want to find a set of cycles keys $\mathcal{K}^j=\{K_{C_{n-1,1}}^j,...,K_{C_{n-1,r}}^j\}$ that encrypted the set of $s$ messages $\mathcal{P}_0=\{P_{00},...,P_{0(s-1)}\}$ in the ciphertexts $\mathcal{P}_r=\{P_{r0},...,P_{r(s-1)}\}$ using the elastic cipher $E_{n-1}$ working with blocks of length $2^{n-1}L$.\\ 
First of all notice that we put an exponent $j$ to $\mathcal{K}$, because there may be more than one set of cycle keys encrypting the same couples $(P_{0i},P_{ri})$ for $i=0,\ldots,s-1$.
One set is sufficient for us.
To find the set of keys $\mathcal{K}^j$ build an attack as follows.\\
Consider the sets: $\mathcal{P}_0'=\{P_{00}\|0...0, ..., P_{0(s-1)}\|0...0\}$ and $\mathcal{P}_r'= \{P_{r0}\| v_{01}...v_{0y}, ..., \\P_{r(s-1)} \| v_{(s-1)1}...v_{(s-1)y}\}$,
where $P_{00}\|0...0$ is the bit string $P_{0i}$ followed by $y$ zeros (we could choose any fixed string of $y$ bits instead of all zeros), and $P_{ri}\| v_{i1}...v_{iy}$ is the bit string $P_{ri}$ followed by any string of $y$ bits.\\
The elements of the sets $\mathcal{P}_0'$ and $\mathcal{P}_r'$ fit as input and output for a $r$-rounds reduced elastic extension $E_n$ of $E_{n-1}$.\\
So we can apply the attack $\mathcal{A}_n$ to the set of pairs $(P_{0i}\|0...0,P_{ri}\| v_{i1}...v_{iy})$ and obtain a set $\mathcal{K}'$ of $t$ sets $\mathcal{K}^{i'}$ $(i=1,\ldots,t)$ of $r$ round keys for $E_n$ in polynomial time in $L$ and $r$.
View this in Fig.\ref{fig12}.
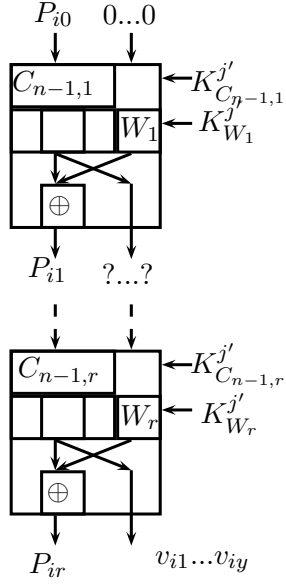
\begin{figure}
\centering
\scalebox{1} 
{
\begin{pspicture}(0,-3.8692188)(5.2490625,3.8692188)
\psframe[linewidth=0.04,dimen=outer](1.89,3.0157812)(0.49,2.4157813)
\psframe[linewidth=0.04,dimen=outer](1.89,2.4157813)(0.49,1.8157812)
\psframe[linewidth=0.04,dimen=outer](1.49,2.4157813)(0.89,1.8157812)
\psframe[linewidth=0.04,dimen=outer](1.49,1.4157813)(0.89,0.81578124)
\psframe[linewidth=0.04,dimen=outer](2.49,2.4157813)(1.89,1.8157812)
\psframe[linewidth=0.04,dimen=outer](2.49,3.0157812)(0.49,0.81578124)
\psline[linewidth=0.04cm,arrowsize=0.05291667cm 2.0,arrowlength=1.4,arrowinset=0.4]{<-}(1.09,3.0157812)(1.09,3.4157813)
\psline[linewidth=0.04cm,arrowsize=0.05291667cm 2.0,arrowlength=1.4,arrowinset=0.4]{<-}(2.09,3.0157812)(2.09,3.4157813)
\psline[linewidth=0.04cm,arrowsize=0.05291667cm 2.0,arrowlength=1.4,arrowinset=0.4]{<-}(1.09,1.4157813)(1.09,1.8157812)
\psline[linewidth=0.04cm,arrowsize=0.05291667cm 2.0,arrowlength=1.4,arrowinset=0.4]{<-}(1.09,0.41578126)(1.09,0.81578124)
\psline[linewidth=0.04cm,arrowsize=0.05291667cm 2.0,arrowlength=1.4,arrowinset=0.4]{<-}(1.09,1.4157813)(2.09,1.8157812)
\psline[linewidth=0.04cm,arrowsize=0.05291667cm 2.0,arrowlength=1.4,arrowinset=0.4]{<-}(2.09,1.4157813)(1.09,1.8157812)
\psline[linewidth=0.04cm,arrowsize=0.05291667cm 2.0,arrowlength=1.4,arrowinset=0.4]{<-}(2.09,0.41578126)(2.09,1.4157813)
\psline[linewidth=0.04cm,arrowsize=0.05291667cm 2.0,arrowlength=1.4,arrowinset=0.4]{<-}(2.49,2.8157814)(2.89,2.8157814)
\psline[linewidth=0.04cm,arrowsize=0.05291667cm 2.0,arrowlength=1.4,arrowinset=0.4]{<-}(2.49,2.2157812)(2.89,2.2157812)
\psframe[linewidth=0.04,dimen=outer](1.89,-0.7842187)(0.49,-1.3842187)
\psframe[linewidth=0.04,dimen=outer](1.89,-1.3842187)(0.49,-1.9842187)
\psframe[linewidth=0.04,dimen=outer](1.49,-1.3842187)(0.89,-1.9842187)
\psframe[linewidth=0.04,dimen=outer](1.49,-2.3842187)(0.89,-2.9842188)
\psframe[linewidth=0.04,dimen=outer](2.49,-1.3842187)(1.89,-1.9842187)
\psframe[linewidth=0.04,dimen=outer](2.49,-0.7842187)(0.49,-2.9842188)
\psline[linewidth=0.04cm,linestyle=dashed,dash=0.16cm 0.16cm,arrowsize=0.05291667cm 2.0,arrowlength=1.4,arrowinset=0.4]{<-}(1.09,-0.7842187)(1.09,-0.18421875)
\psline[linewidth=0.04cm,linestyle=dashed,dash=0.16cm 0.16cm,arrowsize=0.05291667cm 2.0,arrowlength=1.4,arrowinset=0.4]{<-}(2.09,-0.7842187)(2.09,-0.18421875)
\psline[linewidth=0.04cm,arrowsize=0.05291667cm 2.0,arrowlength=1.4,arrowinset=0.4]{<-}(1.09,-2.3842187)(1.09,-1.9842187)
\psline[linewidth=0.04cm,arrowsize=0.05291667cm 2.0,arrowlength=1.4,arrowinset=0.4]{<-}(1.09,-3.3842187)(1.09,-2.9842188)
\psline[linewidth=0.04cm,arrowsize=0.05291667cm 2.0,arrowlength=1.4,arrowinset=0.4]{<-}(1.09,-2.3842187)(2.09,-1.9842187)
\psline[linewidth=0.04cm,arrowsize=0.05291667cm 2.0,arrowlength=1.4,arrowinset=0.4]{<-}(2.09,-2.3842187)(1.09,-1.9842187)
\psline[linewidth=0.04cm,arrowsize=0.05291667cm 2.0,arrowlength=1.4,arrowinset=0.4]{<-}(2.09,-3.3842187)(2.09,-2.3842187)
\psline[linewidth=0.04cm,arrowsize=0.05291667cm 2.0,arrowlength=1.4,arrowinset=0.4]{<-}(2.49,-0.9842188)(2.89,-0.9842188)
\psline[linewidth=0.04cm,arrowsize=0.05291667cm 2.0,arrowlength=1.4,arrowinset=0.4]{<-}(2.49,-1.5842187)(2.89,-1.5842187)
\usefont{T1}{ptm}{m}{n}
\rput(1.0645312,3.6607811){$P_{i0}$}
\usefont{T1}{ptm}{m}{n}
\rput(2.0645313,3.6607811){$0...0$}

\usefont{T1}{ptm}{m}{n}
\rput(1.0645312,2.7207813){$C_{n-1,1}$}
\usefont{T1}{ptm}{m}{n}
\rput(3.5045313,2.7407813){$K_{C_{n-1,1}}^{j'}$}
\usefont{T1}{ptm}{m}{n}
\rput(1.1345313,-1.0792187){$C_{n-1,r}$}
\usefont{T1}{ptm}{m}{n}
\rput(3.5045314,-0.99921876){$K_{C_{n-1,r}}^{j'}$}
\usefont{T1}{ptm}{m}{n}
\rput(3.3545313,2.2407813){$K_{W_1}^{j'}$}
\usefont{T1}{ptm}{m}{n}
\rput(3.3845313,-1.6392188){$K_{W_r}^{j'}$}
\usefont{T1}{ptm}{m}{n}
\rput(2.2045312,2.1407812){$W_1$}
\usefont{T1}{ptm}{m}{n}
\rput(1.1445312,1.1207813){$\oplus$}
\usefont{T1}{ptm}{m}{n}
\rput(2.2045313,-1.6792188){$W_r$}
\usefont{T1}{ptm}{m}{n}
\rput(1.1245313,-2.6792188){$\oplus$}
\usefont{T1}{ptm}{m}{n}
\rput(0.9845312,0.28078124){$P_{i1}$}
\usefont{T1}{ptm}{m}{n}
\rput(2.0445313,0.22078125){$?...?$}
\usefont{T1}{ptm}{m}{n}
\rput(3.0245314,-3.5992188){$v_{i1}...v_{iy}$}
\usefont{T1}{ptm}{m}{n}
\rput(0.9945313,-3.6392188){$P_{ir}$}
\end{pspicture} 
}

\caption{Keys used in $E_n$.
\label{fig12}}
\end{figure}
From $\mathcal{K}'$ we choose one of the possible set of round keys, namely $\mathcal{K}^{j'}=\{K_{C_1}^{j'}\| K_{W_1}^{j'},...,K_{C_r}^{j'}\| K_{W_1}^{j'}\}$.
We can find $K_{C_{n-1,1}}^j$ in the following way.
The last operation in $C_{n-1,1}$ is a whitening of $2^{n-1}L$ bits.
Call this key bit string $K_{C_{n-1,1}[0, \ldots, 2^{n-1}L-1]}^j$, then set:
\begin{multline}
 K_{C_{n-1,1}[0,\ldots,2^{n-1}L-1]}^j=\\
K_{C_{n-1,1}[0,\ldots,2^{n-1}L-1]}^{j'}
\oplus
(0 \| ...
\| 0 \| K_{W_1[2^{n-1}L,...,2^{n-1}L+y-1]}^{j'} \| 0 \| ...
\| 0)
\end{multline}
\noindent Where the number of zeros depends on the position in which the swap occur.
See Figure \ref{fig13}.
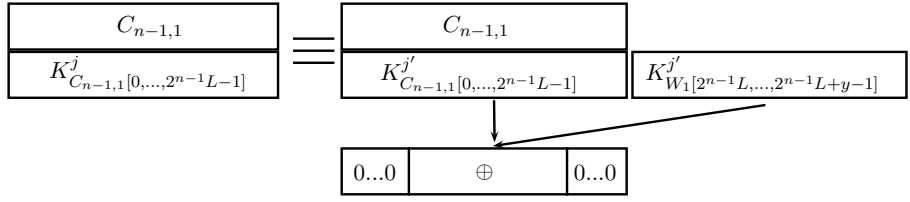
\begin{figure}
\centering
\scalebox{0.8} 
{
\begin{pspicture}(0,-1.62)(17.559063,1.6)
\psframe[linewidth=0.04,dimen=outer](5.4488015,1.6)(0.98,0.8)
\psframe[linewidth=0.04,dimen=outer](5.4488015,0.8)(0.98,0.0)
\psframe[linewidth=0.04,dimen=outer](11.18,1.6)(6.4554687,0.8)
\psframe[linewidth=0.04,dimen=outer](11.18,0.8)(6.4554687,0.0)
\psframe[linewidth=0.04,dimen=outer](15.78,0.8)(11.235469,0.0)
\psframe[linewidth=0.04,dimen=outer](11.18,-0.8)(6.4554687,-1.6)
\psline[linewidth=0.04cm,arrowsize=0.05291667cm 2.0,arrowlength=1.4,arrowinset=0.4]{<-}(8.988802,-0.72)(8.98,-0.02)
\psline[linewidth=0.04cm,arrowsize=0.05291667cm 2.0,arrowlength=1.4,arrowinset=0.4]{<-}(8.96,-0.78)(13.44,-0.1)
\psline[linewidth=0.04cm](5.675469,0.6)(6.3554688,0.6)
\psline[linewidth=0.04cm](5.675469,0.8)(6.3554688,0.8)
\psline[linewidth=0.04cm](5.675469,1.0)(6.3554688,1.0)
\usefont{T1}{ptm}{m}{n}
\rput(3.3245313,1.185){$C_{n-1,1}$}
\usefont{T1}{ptm}{m}{n}
\rput(3.2945313,0.425){$K_{C_{n-1,1}[0,...,2^{n-1}L-1]}^j$}
\usefont{T1}{ptm}{m}{n}
\rput(8.734532,0.405){$K_{C_{n-1,1}[0,...,2^{n-1}L-1]}^{j'}$}
\usefont{T1}{ptm}{m}{n}
\rput(13.394532,0.405){$K_{W_1[2^{n-1}L,...,2^{n-1}L+y-1]}^{j'}$}
\usefont{T1}{ptm}{m}{n}
\rput(8.684531,1.185){$C_{n-1,1}$}
\usefont{T1}{ptm}{m}{n}
\rput(8.824531,-1.215){$\oplus$}
\usefont{T1}{ptm}{m}{n}
\rput(10.644531,-1.215){$0...0$}
\usefont{T1}{ptm}{m}{n}
\rput(7.024531,-1.215){$0...0$}
\psline[linewidth=0.04cm](7.58,-0.8)(7.58,-1.6)
\psline[linewidth=0.04cm](10.18,-0.8)(10.18,-1.6)
\end{pspicture} 
}
\caption{Round key conversion.
\label{fig13}}
\end{figure}

Once we got $K_{C_{n-1,1}[0,...,2^{n-1}L-1]}^j$ we put the remaining bits of $K_{C_{n-1,1}}^j$ equal to the remaining bits of $K_{C_{n-1,1}}^{j'}$, obtaining the desired cycle key $K_{C_{n-1,1}}^j$.
With this key we determine:
\begin{equation}
 P_{i1}=C_{n-1,1}(K_{C_{n-1,1}}^j,P_{i0})
\end{equation}
where $C_{n-1,1}(K_{C_{n-1,1}}^j,.)$ is the cycle function obtained using the cycle key $K_{C_{n-1,1}}^j$.\\
Notice that we can sum the last $y$ key bits of the whitening step because the message to which they are xored is known (all zeros in our case); this is the reason why we can not repeat the same trick with the next cycle key but we have to restart a new attack as we are going to explain.\\
Now we repeat the process mounting the attack to a $r-1$ round reduced extension $E_n$ of $E_{n-1}$, using the couples $(P_{1i}\|0...0,P_{ri}\| v_{i1}...v_{iy})$ and we find one by one the cycle keys $K_{C_{n-1,1}}^j$, with $i=2,\ldots,r$.
The attack is done applying $\mathcal{A}_n$ $r$ times with $s$ pairs of (plaintext, ciphertext) and running one cycle of $E_{n-1}$ for each pair, for a total of $sr$ application of $C_{n-1}$.
If $\mathcal{A}_n$ requires to know the output of each round of the reduced round version of $E_n$, then $\frac{sr(r+1)}{2}$ rounds of $E_n$ have to be computed.
If we call $t_{\mathcal{A}_n}$ the time to run the attack, then the total number of operation required for the reduction is $O(sr^2+rt_{\mathcal{A}_n})$.
In her theorem, Cook also considers the time to check that an expanded key found by $\mathcal{A}_n$ adheres to the key schedule of $E_0$.
For us it is not necessary since the key schedule of $E_{n-1}$ is a pseudorandom generator and as such, any bit sequence must fit.
\end{proof}


\bibliographystyle{acm}
\bibliography{elastic}

\end{document}